\documentclass[11pt,a4paper]{article}
\usepackage{amsmath}  
\usepackage{amsthm}  
\usepackage{amsfonts}
\usepackage{times} 
\usepackage{epsfig}
\usepackage{hyperref}
\hypersetup{colorlinks=true,citecolor=blue, linkcolor=blue, urlcolor=blue}

\def\CSP{{\rm CSP}} 
\def\NCSP{{\rm \#CSP}} 
\def\A{{\bf A}} 
\def\B{{\bf B}} 
\def\cC{{\mathcal C}} 
\def\cD{{\mathcal D}} 
\def\zN{{\mathbb N}}
\def\vc#1#2{#1 _1\zd #1 _{#2}}
\def\zd{,\ldots,} 
\let\dl=\delta 
\let\ve=\varepsilon
\let\vf=\varphi 
\let\vr=\varrho
\let\tm=\times
\let\sse=\subseteq 
\newcommand{\nat}{\mathbb{N}}
\newcommand{\clique}{\textsc{Clique}}
\newcommand{\nclique}{\#\clique}
\newcommand{\tuple}[1]{\mathbf{#1}}
\newcommand{\ar}[1]{\text{ar}(#1)}
\newcommand{\tw}[1]{\text{$tw$}(#1)}
\newcommand{\ignore}[1]{}
\def\Pr{\mathop{{}\rm Pr}\nolimits}
\def\H{{\bf H}} 
 
\newcommand{\prob}[3]{
  \vbox{
  \begin{description}
    \item[\bf Name:] #1
    \vspace{-1.75ex}
    \item[\bf Input:] #2  
    \vspace{-1.75ex}
    \item[\bf Output:] #3
  \end{description}
}}
\newcommand{\probfpt}[4]{
\vbox{
  \begin{description}
    \item[\bf Name:] #1
    \vspace{-1.75ex}
    \item[\bf Input:] #2  
    \vspace{-1.75ex}
    \item[\bf Parameter:] #3  
    \vspace{-1.75ex}
    \item[\bf Output:] #4
  \end{description}
}}

\title{Approximate counting CSP seen from the other side\thanks{An extended
abstract of this work appeared in the \emph{Proceedings of the 44th
International Symposium on Mathematical Foundations of Computer Science
(MFCS'19)}~\cite{Bulatov19:mfcs}. Andrei Bulatov was supported by an NSERC Discovery grant. Stanislav
\v{Z}ivn\'y was supported by a Royal Society University Research Fellowship.
This project has received funding from the European Research Council (ERC) under
the European Union's Horizon 2020 research and innovation programme (grant
agreement No 714532). The paper reflects only the authors' views and not the
views of the ERC or the European Commission. The European Union is not liable
for any use that may be made of the information contained therein.}}

\begin{document}

\newtheorem{theorem}{Theorem}
\newtheorem{corollary}[theorem]{Corollary}
\newtheorem{prop}[theorem]{Proposition} 
\newtheorem{problem}[theorem]{Problem}
\newtheorem{lemma}[theorem]{Lemma} 
\newtheorem{remark}[theorem]{Remark}
\newtheorem{observation}[theorem]{Observation}
\theoremstyle{plain}
\newtheorem*{lemma*}{Lemma}

\author{
Andrei A. Bulatov\\
School of Computing Science, Simon Fraser University, Canada\\
\texttt{abulatov@sfu.ca}
\and
Stanislav \v{Z}ivn\'{y}\\
Department of Computer Science, University of Oxford, UK\\
\texttt{standa.zivny@cs.ox.ac.uk}
}
\date{}
\maketitle

\begin{abstract}
In this paper we study the complexity of counting Constraint Satisfaction
Problems (CSPs) of the form $\NCSP(\cC, -)$, in which the goal is, given a relational
structure $\A$ from a class $\cC$ of structures and an arbitrary structure $\B$,
to find the number of homomorphisms from $\A$ to $\B$. Flum and Grohe showed 
that $\NCSP(\cC, -)$ is solvable in polynomial time if
$\cC$ has bounded treewidth [FOCS'02]. Building on the work of Grohe~[JACM'07]
on decision CSPs, 
Dalmau and Jonsson then showed that, if $\cC$ is a recursively enumerable class 
of relational structures of bounded arity, then assuming FPT
$\ne$ \#W[1], there are no other cases of $\NCSP(\cC, -)$
solvable exactly in polynomial time (or even fixed-parameter time)~[TCS'04].

We show that, assuming FPT $\ne$ W[1] (under randomised parameterised reductions)
and for $\cC$ satisfying certain general conditions, $\NCSP(\cC,-)$ is not solvable 
even \emph{approximately} for $\cC$ of unbounded treewidth; that is,
there is no fixed parameter tractable (and thus also not fully polynomial) randomised 
approximation scheme for $\NCSP(\cC, -)$. In particular, our condition
generalises the case when $\cC$ is closed under taking minors.
\end{abstract}

\section{Introduction}
\label{sec:intro}

The Constraint Satisfaction Problem (CSP) asks to decide the existence of a 
homomorphism between two given relational structures (or to find the number 
of such homomorphisms). It has been used to model a vast variety of combinatorial
problems and has attracted much attention. Since the general CSP is NP-complete
(\#P-complete in the counting case) and because one needs to model specific 
computational problems, various restricted versions of the CSP have been 
considered. More precisely, let $\cC$ and $\cD$ be two classes of relational 
structures. In this paper we will assume that structures from $\cC,\cD$ only 
have predicate symbols of bounded arity. The \emph{constraint satisfaction 
problem} (CSP) parameterised by $\cC$ and $\cD$ is
the following computational problem, denoted by $\CSP(\cC, \cD)$: given
$\A\in\cC$ and $\B\in\cD$, is there a homomorphism from $\A$ to $\B$? 
CSPs in which both input structures are restricted have not received
much attention (with a notable exception of matrix partitions 
\cite{Feder05:matrix,Feder07:matrix} and assorted graph problems
on restricted classes of graphs). However, the two most natural restrictions
have been intensively studied over the last two decades.
Let $-$ denote the class of all (bounded-arity) relational structures, or, equivalently,
indicate that there are no restrictions on the corresponding input structure.

Problems of the form $\CSP(-, \{\B\})$, where $\B$ is a fixed finite relational 
structure, are known
as \emph{nonuniform} or \emph{language-restricted} CSPs~\cite{Kolaitis00:jcss}. For 
instance, if $ \B=K_3$ is the complete graph
on $3$ vertices then $\CSP(-, \{\B\})$ is the standard {\sc 3-Colouring}
problem~\cite{Garey79:intractability}. The study of nonuniform CSPs has been 
initiated by Schaefer \cite{Schaefer78:complexity} who considered the case 
of $\CSP(-, \{\B\})$ for 2-element structures $\B$. The complexity of 
$\CSP(-, \{\H\})$, for a fixed graph $\H$, was studied 
under the name of $\H$-colouring by Hell and Ne\v{s}et\v{r}il 
\cite{Hell90:h-coloring}.
General nonuniform CSPs have been studied extensively since the seminal paper 
of Feder and Vardi~\cite{Feder98:monotone} who in particular proposed the so-called
Dichotomy Conjecture stating that every nonuniform CSP is either solvable in 
polynomial time or is NP-complete. The complexity of nonuniform CSPs has been 
resolved only recently in two independent papers by Bulatov~\cite{Bulatov17:focs} 
and Zhuk~\cite{Zhuk17:focs}, which confirmed the dichotomy conjecture of 
Feder and Vardi and also its algebraic version~\cite{Bulatov05:classifying}.

CSPs restricted on the other side, that is, of the form $\CSP(\cC, -)$, where $\cC$ 
is a fixed (infinite) class of finite relational
structures, are known as \emph{structurally-restricted} CSPs. 
For instance, if $\cC=\cup_{k\ge 1}\{K_k\}$ is the class of cliques of all
sizes then $\CSP(\cC, -)$ is the standard {\clique} 
problem~\cite{Garey79:intractability}. In this case the complexity of CSPs is
related to various ``width'' parameters of the associated class of graphs. 
For a relational structure $\A$ let $G(\A)$ denote the Gaifman graph of
$\A$, that is, the graph whose vertices are the elements of $\A$, and 
vertices $v,w$ are connected with an edge whenever $v$ and $w$ occur in
the same tuple of some relation of $\A$. Then $G(\cC)$ denotes the class of
Gaifman graphs of structures from $\cC$, and we refer to the treewidth of 
$G(\A)$ as the treewidth of $\A$.
Dalmau, Kolaitis, and Vardi showed that $\CSP(\cC, -)$ is in PTIME if $\cC$
has bounded treewidth modulo homomorphic equivalence~\cite{Dalmau02:width}.
Grohe then showed that, assuming FPT $\ne$ W[1], there are no other cases of
(bounded arity) $\CSP(\cC, -)$ solvable in polynomial time (or even
fixed-parameter time, where the parameter is the size of the left-hand side
structure)~\cite{Grohe07:jacm}. The case of structures with unbounded
arity was extensively studied by Gottlob et al.\ who introduced the concept
of bounded hypertree width in an attempt to characterise structurally
restricted CSPs solvable in polynomial time~\cite{Gottlob02:jcss-hypertree}. The search for a right condition
is still going on, and the most general structural property that guarantees that 
$\CSP(\cC, -)$ is solvable in polynomial time is fractional hypertree width
introduced by Grohe and Marx \cite{Grohe14:talg-fractional}. 
Finally, Marx showed that the most general condition, assuming the
exponential-time hypothesis, that captures structurally-restricted CSPs solvable
in fixed-parameter time is that of submodular width~\cite{Marx13:jacm}.

An important problem related to the CSP is counting: Given a CSP instance,
that is, two relational structures $\A$ and $\B$, find the number of homomorphisms
from $\A$ to $\B$. We again consider restricted versions of this problem. 
More precisely, for two classes $\cC$ and $\cD$ of relational structures, 
$\NCSP(\cC, \cD)$ denotes the following computational problem:
 given $\A\in\cC$ and $\B\in\cD$, how many homomorphisms are there 
 from $\A$ to $\B$? This problem is referred to as a \emph{counting} CSP.
 Similar to decision CSPs, problems of the form $\NCSP(-, \cD)$ and
 $\NCSP(\cC, -)$ are the two most studied ways to restrict the counting
 CSP, and the research on these problems follows a similar pattern as their
 decision counterparts. 

For a fixed finite relational structure $\B$, the complexity of the nonuniform
problem $\NCSP(-, \{\B\})$ was characterised for graphs by Dyer and 
Greenhill \cite{Dyer00:counting} and for 2-element structures by Creignou and
Hermann \cite{Creignou96:complexity}. The complexity of the general nonuniform 
counting CSPs was resolved by Bulatov~\cite{Bulatov13:jacm-dichotomy} and 
Dyer and Richerby \cite{Dyer13:sicomp-csp}. As in the case of the decision version the
complexity of nonuniform counting CSPs is determined by their algebraic properties,
and every such CSP is either solvable in polynomial time or is \#P-complete.
These dichotomy results were later extended to the case of weighted counting
CSP, for which Cai and Chen obtained a complexity classification of counting 
CSPs with complex weights~\cite{Cai17:jacm}.

The complexity of counting CSPs with restrictions on the left hand side structures 
also turns out to be related to treewidth. Flum and Grohe showed that 
$\NCSP(\cC, -)$ is solvable in polynomial time if $\cC$ has bounded
treewidth~\cite{Flum02:focs}. Dalmau and Jonsson then showed that, assuming FPT
$\ne$ \#W[1], there are no other cases of (bounded arity) $\NCSP(\cC, -)$
solvable exactly in polynomial time (or, again, even fixed-parameter
time)~\cite{Dalmau04:side}. Note that the result of Dalmau and Jonsson states that
the class $\cC$ itself has to be of bounded treewidth, while in Grohe's 
characterisation of polynomial-time solvable decision CSPs of the form $\CSP(\cC,-)$
it is the class of cores of structures from $\cC$ that has to have bounded
treewidth. There has also been some research on 
counting problems over structures of unbounded arity. First, it was showed that
notions sufficient for polynomial-time solvability of decision CSPs can be
lifted to the problem of counting CSPs. In particular, the polynomial-time
solvability of $\NCSP(\cC,-)$ was shown by Pichler and Skritek
for $\cC$ of bounded hypertree width~\cite{Pichler13:jcss}, by Mengel for $\cC$
of bounded fractional hypertree width~\cite{Menge13:phd}, and finally by
Farnqvist for $\cC$ of bounded
submodular width~\cite{Farnqvist13:phd}. Secondly, the work of Brault-Baron et al. showed that the (unbounded arity)
structurally-restricted $\NCSP(\cC,-)$ are solvable in polynomial time for the
class $\cC$ of $\beta$-acyclic hypergraphs~\cite{Brault-baron15:stacs}.\footnote{Brault-Baron
et al.\ \cite{Brault-baron15:stacs} show their tractability results for so-called CSPs with default values,
which in particular includes $\NCSP(\cC,-)$ as defined here.}

The results we have mentioned so far concern exact counting; however,
many applications of counting problems allow for approximation algorithms as well.
For nonuniform CSPs the complexity landscape is much more complicated
than the dichotomy results for decision CSPs or exact counting. The analogue
of ``easily solvable'' problems in this case are those that admit a Fully Polynomial 
Randomised Approximation Scheme (FPRAS): a randomised algorithm that, given an instance
and an error tolerance $\ve\in(0,1)$ returns in time polynomial in the size of the instance
and $\ve^{-1}$ a result which is with high probability a multiplicative
$(1+\ve)$-approximation of 
the exact solution. The parameterised version of this algorithmic model is known
as a Fixed Parameter Tractable Randomised Approximation Scheme (FPTRAS). However, unlike exact counting or the decision CSP, it is not very likely there is a concise and clear complexity classification. For instance,  
Dyer et al.\ \cite{Dyer04:algorithmica} identified a sequence of counting CSPs, Bipartite $q$-Colouring, that are likely to attain an infinite hierarchy of 
approximation 
complexities. Only a handful of results exist for 
the approximation complexity of counting nonuniform CSPs. The approximation 
complexity of $\NCSP(-,\{\B\})$ for 2-element structures $\B$ was characterised 
by Dyer et al.\ \cite{Dyer10:trichotomy}, where a trichotomy theorem was proved: 
for every
2-element structure $\B$ the problem $\NCSP(-,\{\B\})$ either admits an FPRAS,
or is interreducible with \#SAT or with the problem \#BIS of counting independent
sets in bipartite graphs. Apart from this only partial results are known. If
$\B$ is a connected graph and $\NCSP(-,\{\B\})$ does not admit an FPRAS, 
then Galanis, Goldberg and Jerrum \cite{Galanis16:h-coloring} showed that
$\NCSP(-,\{\B\})$ is at least as hard as \#BIS. Also, if every unary relation is a 
part of $\B$ a complexity classification of $\NCSP(-,\{\B\})$ can be extracted
from the results of Chen et al.\ \cite{Chen15:conservative},\footnote{Chen et al.\ 
\cite{Chen15:conservative} studied the weighted version of $\NCSP(-,\{\B\})$,
and although their result does not provide a complete characterisation of 
the weighted problem, it allows to determine the complexity of $\NCSP(-,\{\B\})$ 
as defined here.} see also~\cite{Galanis17:toct}.

\paragraph*{Our Contribution} 
It should be clear by now that the picture painted by the short survey above
misses one piece: the approximation complexity of structurally restricted CSPs. 
This is the main contribution of this paper. 

Let $\cC$ be a class of bounded-arity relational structures. If the treewidth of
$\cC$ modulo homomorphic equivalence is unbounded then, by Grohe's
result~\cite{Grohe07:jacm}, it is hard to test for the existence of a
homomorphism from $\A$ to $\B$, where $\A\in\cC$, for any instance $\A,\B$ of
$\CSP(\cC,-)$. Using standard techniques (see, e.g., the proof
of~\cite[Proposition~3.16]{Meeks16:dam}), this implies, assuming that FPT $\ne$
\#W[1] (under randomised parameterised reductions~\cite{Downey98:tcs}), that
there is not an FPTRAS for
$\NCSP(\cC, -)$, let alone an FPRAS.
Consequently, the tractability boundary for approximate counting of
$\NCSP(\cC, -)$ lies between bounded treewidth and bounded treewidth modulo
homomorphic equivalence.

As our main result, we show that for $\cC$ such that a certain class of graphs
(to be defined later) is a subset of $G(\cC)$, $\NCSP(\cC,-)$ cannot be solved 
even \emph{approximately} for $\cC$ of unbounded treewidth, assuming FPT 
$\ne$ W[1] (under randomised parameterised reductions). Before we introduce 
the classes of graphs we use, we review how the hardness of $\CSP(\cC, -)$ or
$\NCSP(\cC, -)$ is usually proved.

We follow the hardness proof of Grohe for decision CSPs~\cite{Grohe07:jacm},
which was lifted to exact counting CSPs by Dalmau and
Jonsson~\cite{Dalmau04:side}. In fact Grohe's result had an important
precursor~\cite{Grohe01:evaluation}. The key idea is a reduction from the
parameterised {\clique} problem to $\CSP(\cC, -)$. Let $G=(V,E)$ and $k$ be an
instance of the $p$-{\clique} problem, where $k$ is the parameter. Broadly
speaking, the reduction works as follows. For a class of unbounded treewidth,
the Excluded Grid Theorem of Robertson and Seymour~\cite{Robertson86:excluding}
guarantees the existence of the $(k\tm{k\choose 2})$-grid (as a minor of some
structure $\A\in\cC$), which is used to encode the existence of a $k$-clique in
$G$ as a certain structure $\B$. The encoding usually means that $G$ has a
$k$-clique if and only if there is a homomorphism from $\A$ to $\B$ whose image covers a copy of the grid built in $\B$. For decision CSPs, the correctness of the
reduction --- that there are no homomorphisms from $\A$ to $\B$ not satisfying this condition --- is
achieved by dealing with coloured grids~\cite{Grohe01:evaluation} or by dealing
with structures whose cores have unbounded treewidth (with another complication
caused by minor maps)~\cite{Grohe07:jacm}. For the complexity of exact counting
CSPs, the correctness of the reduction \cite{Dalmau04:side} is achieved by 
employing interpolation or the inclusion-exclusion principle, a common tool in 
exact counting.

None of these two methods can be applied to approximate solving 
$\NCSP(\cC,-)$. We cannot assume that the class of cores of $\cC$ has unbounded treewidth, because then 
by \cite{Grohe07:jacm} even the decision problem cannot be solved in polynomial
time, which immediately rules out the existence of an FPRAS. Interpolation 
techniques such as the inclusion-exclusion principle are also well known to be incompatible with
approximate counting. The standard tool in approximate counting to achieve 
the same goal of prohibiting homomorphisms except ones from a certain restricted 
type, is to use gadgets to amplify the number of homomorphisms of the required 
type. We give a reduction from $p$-{\nclique} to $\NCSP(\cC, -)$ by using
``fan-grids'', formally introduced in Section~\ref{sec:fan}. Unfortunately, 
due to the delicate nature of approximation preserving reductions, we cannot use
minors and minor maps and have to assume that ``fan-grids'' themselves are
present in $G(\cC)$. 
(In Section~\ref{sec:conclusion}, we will briefly discuss how a weaker
assumption can be used to obtain the same result.)
By the Excluded Grid Theorem~\cite{Robertson86:excluding}, if $\cC$ is \emph{closed under taking minors},
then $G(\cC)$ contains all the fan-grids (details are given in
Section~\ref{sec:fan} and in particular in Lemma~\ref{lem:fan-minor}). Thus, the
classes $\cC$ for which we establish the hardness of $\NCSP(\cC, -)$ includes
all classes $\cC$ that are closed under taking minors.\footnote{We remark that
the hardness for $\cC$ closed under taking minors follows from Grohe's
classification~\cite{Grohe07:jacm} of decision CSPs. Indeed, for $\cC$ of
unbounded treewidth, the Excluded Grid Theorem~\cite{Robertson86:excluding}
gives grids of arbitrary sizes. Since every planar graph is a minor of some
grid~\cite{Diestel10:graph}, $\cC$ contains all planar graphs. As there exist
planar graphs of arbitrarily large treewidth that are also minimal with respect to
homomorphic equivalence, Grohe's result gives W[1]-hardness of $\CSP(\cC, -)$
and hence $\NCSP(\cC, -)$ cannot have an FPRAS/FPTRAS.}

\section{Preliminaries}\label{sec:preliminaries}

$\zN$ denotes the set of positive integers. For every $n\in\zN$, we let
$[n]=\{1,\ldots,n\}$.

\subsection{Relational Structures and Homomorphisms}

A \emph{relational signature} is a finite set $\tau$ of relation symbols $R$,
each with a specified arity $\ar{R}$. A \emph{relational structure} $\A$
over a relational signature $\tau$ (or a $\tau$-structure, for short)
is a finite universe $A$ together with one relation $R^\A\subseteq
A^{\ar{R}}$ for each symbol $R \in \tau$.
The size $\|\A\|$ of a relational structure $\A$ is defined as
\[
  \|\A\|=|\tau|+|A|+\sum_{R\in\tau}|R^\A|\cdot\ar{R}.
\]

Let $R$ be a binary relational symbol. We will sometimes view graphs as $\{R\}$-structures.

A \emph{homomorphism} from a relational $\tau$-structure $\A$ (with universe
$A$) to a relational $\tau$-structure $\B$ (with universe $B$) is a mapping
$\vf:A \to B$ such that for all $R\in\tau$ and all tuples $\tuple{x}\in R^\A$ we
have $\vf(\tuple{x})\in R^\B$. 

Two structures $\A$ and $\B$ are \emph{homomorphically equivalent} if there is a
homomorphism from $\A$ to $\B$ and a homomorphism from $\B$ to $\A$.

Let $\cC$ be a class of relational structures. We say that $\cC$ has \emph{bounded arity} if there is a constant $r\ge 1$ such that for every $\tau$-structure $\A\in\cC$ and
$R\in\tau$, we have that $\ar{R}\leq r$.

\subsection{Treewidth and Minors}

The notion of treewidth, introduced by Robertson and
Seymour~\cite{Robertson84:minors3}, is a well-known measure of the tree-likeness of a
graph~\cite{Diestel10:graph}.
Let $G=(V(G),E(G))$ be a graph. A \emph{tree decomposition} of $G$ 
is a pair $(T,\beta)$ where $T=(V(T),E(T))$ is a tree and $\beta$ is a function that maps 
each node $t\in V(T)$ to a subset of $V(G)$ such that
\begin{enumerate}
\item $V(G)=\bigcup_{t\in V(T)} \beta(t)$, 
\item for every $u\in V(G)$, the set 
$\{t\in V(T)\mid u\in \beta(t)\}$ induces a connected subgraph of $T$, and 
\item for every edge $\{u,v\}\in E(G)$, there is a node $t\in V(T)$ with $\{u,v\}\subseteq \beta(t)$. 
\end{enumerate} 
The \emph{width} of the decomposition $(T,\beta)$ is $\max\{|\beta(t)|\mid t\in V(T)\}-1$. 
The \emph{treewidth} $\tw{G}$ of a graph $G$ is the minimum width over all its tree decompositions. 

Let $\A$ be a relational structure over relational signature $\tau$. 
The \emph{Gaifman graph} (also known as \emph{primal graph}) of $\A$,
denoted by $G(\A)$, is the graph whose vertex set is the
universe of $\A$ and whose edges are the
pairs $(u,v)$ for which there is a tuple $\tuple{x}$ and a relation symbol $R\in \tau$ 
such that $u,v$ appear in $\tuple{x}$ and $\tuple{x}\in R^{\A}$. 

Let $\cC$ be a class of relational structures. We say that $\cC$ has
\emph{bounded treewidth} if there exists $w\ge 1$ such that
$\tw{\A}=\tw{G(\A)}\leq w$ for every $\A\in\cC$. We say that $\cC$ has
\emph{bounded treewidth modulo homomorphic equivalence} if there exists $w\ge 1$
such that every $\A\in\cC$ is homomorphically equivalent to $\A'$ with
$\tw{\A'}\leq w$.

A graph $H$ is a \emph{minor} of a graph $G$ if $H$ is isomorphic to a graph
that can be obtained from a subgraph of $G$ by contracting edges (for more
details, see, e.g.,~\cite{Diestel10:graph}). 

For $k,\ell\ge 1$, the $(k\tm\ell)$-grid is the graph with the
vertex set $[k]\tm[\ell]$ and an edge between $(i,j)$ and $(i',j')$ iff
$|i-i'|+|j-j'|=1$. Treewidth and minors are intimately connected via the
celebrated Excluded Grid Theorem of Robertson and Seymour.

\begin{theorem}[\cite{Robertson86:excluding}]\label{thm:excluded}
  For every $k$ there exists a $w(k)$ such that the $(k\tm k)$-grid is a
  minor of every graph of treewidth at least $w(k)$.
\end{theorem}

Let $\cC$ be a class of relational structures. We say that $\cC$ if \emph{closed
under taking minors} if for every $\A\in\cC$ and for every minor $H$ of $G(\A)$,
there is a structure $\A'\in \cC$ such that $G(\A')$ is isomorphic to $H$.

\section{Counting CSP}

\subsection{Exact Counting CSP}

Let $\cC$ be a class of relational structures. 
We will be interested in the computational complexity of the following problem.

\prob
{$\NCSP(\cC, -)$}
{Two relational structures $\A$ and $\B$ over the same signature with $\A\in\cC$.}
{The number of homomorphisms from $\A$ to $\B$.}

We say that $\NCSP(\cC, -)$ is in FP, the class of function problems solvable in
\emph{polynomial time}, if there is a deterministic algorithm that solves any
instance $\A,\B$ of $\NCSP(\cC, -)$ in time $(\|\A\|+\|\B\|)^{O(1)}$.

We will also consider the parameterised version of $\NCSP(\cC, -)$.

\probfpt
{$p$-$\NCSP(\cC, -)$}
{Two relational structures $\A$ and $\B$ over the same signature with $\A\in\cC$.}
{$\|\A\|$.}
{The number of homomorphisms from $\A$ to $\B$.}

We say that $p$-$\NCSP(\cC, -)$ is in FPT, the class of problems that are
\emph{fixed-parameter tractable}, if there is a deterministic algorithm that
solves any instance $\A,\B$ of $p$-$\NCSP(\cC, -)$ in time
$f(\|\A\|)\cdot\|\B\|^{O(1)}$, where $f:\mathbb{N} \to \mathbb{N}$ is an
arbitrary computable function. 

The class W[1], introduced in~\cite{Downey95:sicomp}, can be seen as an analogue
of NP in parameterised complexity theory. Proving W[1]-hardness of a problem
(under a parameterised reduction which may be randomised),
is a strong indication that the problem is not solvable in fixed-parameter time
as it is believed that FPT $\neq$ W[1]. For counting problems, \#W[1] is the
parameterised analogue of \#P. Similarly to the belief that FP $\neq$  \#P, it is
believed that FPT $\neq$ \#W[1]. We refer the reader
to~\cite{Flum06:parametrized} for the definitions of W[1] and \#W[1], and for
more details on parameterised complexity in general.

Dalmau and Jonsson established the following result.

\begin{theorem}[\cite{Dalmau04:side}]\label{thm:exact}

Assume FPT $\ne$ \#W[1] under parameterised reductions. Let $\cC$ be a
recursively enumerable class of relational structures of bounded arity. Then,
the following are equivalent: 

\begin{enumerate}
  \item $\NCSP(\cC, -)$ is in FP.
  \item $p$-$\NCSP(\cC, -)$ is in FPT.
  \item $\cC$ has bounded treewidth. 
\end{enumerate}
\end{theorem}

The following problem is an example of a \#W[1]-hard problem, as established by
Flum and Grohe~\cite{Flum04:sicomp}.

\probfpt
{$p$-\nclique}
{A graph $G$ and $k\in\zN$.}
{$k$.}
{The number of cliques of size $k$ in $G$.}

Note that $p$-{\nclique} can be modelled as $p$-$\NCSP(\cC, -)$ if we set $\cC$ to be the set of cliques of all possible sizes.
The decision version of $p$-{\nclique} was shown to be W[1]-hard by Downey and
Fellows~\cite{DF95:fpt}.

\probfpt
{$p$-\clique}
{A graph $G$ and $k\in\zN$.}
{$k$.}
{Decide if $G$ contains a clique of size $k$.}

\subsection{Approximate Counting CSP}

In view of our complete understanding of the exact complexity of $\NCSP(\cC, -)$
for $\cC$ of bounded arity (cf. Theorem~\ref{thm:exact}), we will be interested
in \emph{approximation} algorithms for $\NCSP(\cC, -)$. In particular, are there
any new classes $\cC$ of bounded arity for which the problem $\NCSP(\cC, -)$ can
be solved efficiently (if only approximately)? We will provide a partial answer
to this question (cf. Theorem~\ref{thm:main}): for certain general bounded-arity
classes $\cC$ (which include classes that are \emph{closed under taking
minors}), the answer is no!

The notion of efficiency for approximate counting is that of a fully polynomial
randomised approximation scheme~\cite{Mitzenmacher2017} and its parameterised
analogue, a fixed parameter tractable randomised approximation
scheme, originally introduced by Arvind and Raman~\cite{Arvind02:isaac}. We now
define both concepts.

A \emph{randomised approximation scheme} (RAS) for a function $f:\Sigma^*\to\zN$
is a randomised algorithm that takes as input
$(x,\ve)\in\Sigma^*\tm (0,1)$ and produces as output an integer random
variable $X$ satisfying the condition 
$\Pr(|X-f(x)|\leq \ve f(x))\ge 3/4$. A RAS for a counting problem is
called \emph{fully polynomial} (FPRAS) if on input of size $n$ it runs in time
$p(n,\ve^{-1})$ for some fixed polynomial $p$. A RAS for a parameterised
counting problem is called \emph{fixed parameter tractable} (FPTRAS) if on input
of size $n$ with parameter $k$ it runs in time $f(k)\cdot p(n,\ve^{-1})$,
where $p$ is a fixed polynomial and $f$ is an arbitrary computable function.

To compare approximation complexity of (parameterised) counting problems 
two types of reductions are used. Suppose $f,g:\Sigma^*\to\zN$. 
An \emph{approximation preserving reduction} (AP-reduction) 
\cite{Dyer04:algorithmica} from $f$ to $g$ is a probabilistic oracle
Turing machine $M$ that takes as input a pair 
$(x,\ve)\in\Sigma^*\tm(0,1)$, and satisfies the following three 
conditions: (i) every oracle call made by $M$ is of the form $(w,\dl)$, 
where $w\in\Sigma^*$ is an instance of $g$, and $0<\dl<1$ is an error
bound satisfying $\dl^{-1}\le\mathsf{poly}(|x|,\ve^{-1})$; (ii) the TM 
$M$ meets the specification for being a randomised approximation 
scheme for $f$ whenever the oracle meets the specification for being 
a randomised approximation scheme for $g$; and (iii) the
running time of $M$ is polynomial in $|x|$ and $\ve^{-1}$. 

Similar to \cite{Meeks16:dam} we also use the parameterised version of 
AP-reductions. Again, let $f,g:\Sigma^*\to\zN$. 
A \emph{parameterised approximation preserving reduction} (parameterised AP-reduction)
from $f$ to $g$ is a probabilistic oracle
Turing machine $M$ that takes as input a triple 
$(x,k,\ve)\in\Sigma^*\tm(0,1)$, and satisfies the following three 
conditions: (i) every oracle call made by $M$ is of the form $(w,k',\dl)$, 
where $w\in\Sigma^*$ is an instance of $g$, $k'\le h(k)$ for some 
computable function $h$, and $0<\dl<1$ is an error
bound satisfying $\dl^{-1}\le\mathsf{poly}(|x|,\ve^{-1})$; (ii) the TM 
$M$ meets the specification for being a randomised approximation 
scheme for $f$ whenever the oracle meets the specification for being 
a randomised approximation scheme for $g$; and (iii) $M$ is fixed-parameter
tractable with respect to $k$ and polynomial in $|x|$ and $\ve^{-1}$. 

\subsection{Main Result}
\label{sec:fan}

The following concept  plays a key role in this paper.
Let $k,r,\ell_1,\ell_2\in\zN$ with $k,r\geq 8$. Intuitively, the \emph{fan-grid}
is a $(k\times r)$-grid with extra degree-one vertices attached to certain
special (called ``fan'') vertices. Formally, the \emph{fan-grid}
$L(k,r,\ell_1,\ell_2)$ 
is a graph with vertex set $L_1\cup L_2$, where $L_1=\{(i,p)\mid i\in[k],
p\in[r]\}$, $L_2=M_1\cup\dots\cup M_{12}$, where $\vc M{12}$ are disjoint and
$|M_i|=\ell_1$ for $i\in[4]$, and $|M_i|=\ell_2$ for $i\in\{5\zd12\}$. 
Vertices from $L_1$ will be called \emph{grid vertices}.
Vertices $u_1=(1,1)$, $u_2=(1,r)$, $u_3=(k,1)$, $u_4=(k,r)$, $u_5=(1,3)$,
$u_6=(1,r-3)$, $u_7=(k,3)$, $u_8=(k,r-3)$, $u_9=(3,1)$, $u_{10}=(4,r)$,
$u_{11}=(k-2,1)$, $u_{12}=(k-3,r)$ will be called \emph{fan vertices}, and
$u_1,u_2,u_3,u_4$ will be called \emph{corner vertices}. The edges of the fan
grid are as follows: $(i,p)(i',p')$ for $|i-i'|+|p-p'|=1$, and $wu_i$ for each
$w\in M_i$ and $i\in[12]$, see Figure~\ref{fig:L-structure}.

\begin{figure}[t]
\centerline{\includegraphics[totalheight=6cm,keepaspectratio]%
{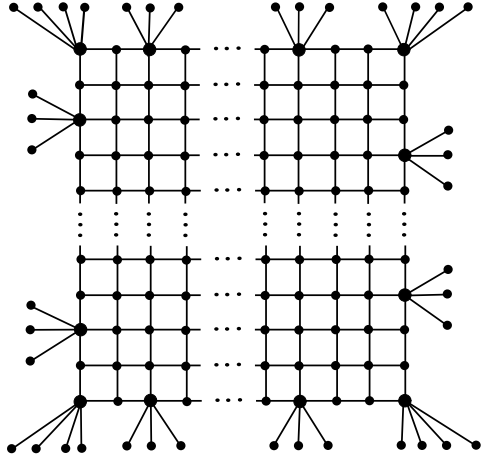}}
\caption{An example of a fan-grid with $\ell_1=4$ and $\ell_2=3$. Fan vertices are shown by larger dots.}\label{fig:L-structure}
\end{figure}

We call a class $\cC$ of relational structures of bounded arity a 
\emph{fan class} if
either $\cC$ has bounded treewidth or 
for any parameters $k,r,\ell_1,\ell_2\in\nat$ we have that $G(\cC)$ contains the
fan-grid $L(k,r,\ell_1,\ell_2)$.

The following is our main result.

\begin{theorem}[\textbf{Main}]\label{thm:main}

Assume FPT $\neq$ W[1] under randomised parameterised reductions. Let $\cC$ be a
recursively enumerable class of relational structures of bounded arity. If
$\cC$ is a fan class then the following are equivalent: 

\begin{enumerate}
\item
$\NCSP(\cC, -)$ is polynomial time solvable.
\item 
  $\NCSP(\cC, -)$ admits an FPRAS. 
\item 
  $p$-$\NCSP(\cC, -)$ admits an FPTRAS. 
\item 
  $\cC$ has bounded treewidth. 
\end{enumerate}
\end{theorem}

Let $\cC$ be a recursively enumerable class of relational structures of bounded
arity and closed under taking minors. We claim that $\cC$ is a fan class and thus
Theorem~\ref{thm:main} applies to such $\cC$. For this we need
Theorem~\ref{thm:excluded}. In particular, for any
$k,r,\ell_1,\ell_2\in\nat$, if $\cC$ is not of bounded treewidth then, by
Theorem~\ref{thm:excluded}, $G(\cC)$ contains an $(s\tm s)$-grid, where
$s=\max(k+2\ell_1,r+2\ell_2)$, and thus also a $(k+2\ell_1)\tm(r+2\ell_2)$-grid.
The following simple lemma
then shows that fan-grids are minors of grids (of appropriate size). 

\begin{lemma}\label{lem:fan-minor}
  $L(k,r,\ell_1,\ell_2)$ is a minor of $(t\tm t')$-grid, where $t=k+2\ell$,
  $t'=r+2\ell$ and $\ell=\max(\ell_1,\ell_2)$.
\end{lemma}

\begin{proof}
Take the subgraph of the $(t\tm t')$-grid as shown in Figure~\ref{fig:fan-minor}
and contract the paths shown with thicker edges.
\end{proof}

\begin{figure}[thb]
\centerline{\includegraphics[totalheight=6cm,keepaspectratio]%
{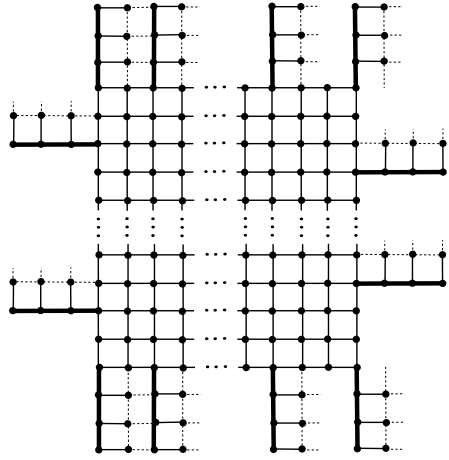}}
\caption{Fan-grid as a minor. In the subgraph of the bigger grid 
shown in solid lines contract the thick edges.}\label{fig:fan-minor}
\end{figure}

\section{Proof of Theorem~\ref{thm:main}}
\label{sec:reduction}

Conditions (1) and (4) in Theorem~\ref{thm:main} are equivalent by
Theorem~\ref{thm:exact}. Implications ``(1) $\Rightarrow$ (2) $\Rightarrow$
(3)'' are trivial. Our main contribution is to prove the ``(3) $\Rightarrow$
(4)'' implication. 

The main idea of the proof is as follows. Assuming that $\NCSP(\cC, -)$ admits
an FPTRAS for a fan class $\cC$, we will demonstrate that $\cC$ has bounded
treewidth. For the sake of contradiction, assume that $\cC$ has unbounded
treewidth. We will exhibit a parameterised reduction from $p$-{\nclique} to
$p$-$\NCSP(\cC, -)$, which gives an FPTRAS for $p$-{\nclique} assuming an FPTRAS
for $p$-$\NCSP(\cC, -)$. Given a graph $G$ and an integer $k$, our reduction
builds (in Section~\ref{sec:construction}) a graph $H=H(G,k,W_1,W_2)$ in such a
way that the number of homomorphisms from a fan-grid $L=L(k,r,\ell_1,\ell_2)$
(defined in Section~\ref{sec:fan}) to $H$ approximates the number of $k$-cliques
in $G$. Section~\ref{subsec:weights} gives details on the number of possible
homomorphisms from $H$ to $L$. (The numbers $\ell_1$, $\ell_2$, $W_1$, and $W_2$
are carefully chosen to make the reduction work.) Section~\ref{subsec:pieces}
then fits the pieces together and describes the reduction in detail.

\subsection{Construction}
\label{sec:construction}

Let $G=(V,E)$ be a graph with $n=|V|$ and $m=|E|$. Let $k\in\zN$. 
We construct a graph $H(G,k,W_1,W_2)$ for
$W_1,W_2>2(n+m)$ as follows.\footnote{A similar if slightly simpler construction
is described and illustrated in~\cite[Section~4.1.1]{Meeks16:dam}.}
Let $r={k\choose2}$ and let $\vr$ be a correspondence 
between $[r]$ and the set of 2-element sets 
$\{\{i,j\}\mid i,j\in[k], i\ne j\}$. For $i\in[k]$ and $p\in[r]$, we write
$i\in p$ as a shorthand for $i\in\vr(p)$. The vertex set of $H(G,k,W_1,W_2)$ is the union of two sets
$H_1\cup H_2$, defined by
\begin{align*}
H_1 &= \{(v,e,i,p)\mid v\in V, e\in E, \text{ and } v\in e\iff i\in p\},\\
H_2 &= K_1\cup\dots\cup K_{12},
\end{align*}
where $\vc K{12}$ are disjoint and 
$|K_i|=W_1$ for $i\in[4]$, $|K_i|=W_2$ for $i\in\{5\zd12\}$.

As in fan-grids, vertices of the form $(v,e,1,1)$, $(v,e,1,r)$, $(v,e,k,1)$,
$(v,e,k,r)$, $(v,e,1,3)$, $(v,e,1,r-3)$, $(v,e,k,3)$, $(v,e,k,r-3)$,
$(v,e,3,1)$, $(v,e,4,r)$, $(v,e,k-2,1)$, $(v,e,k-3,r)$ will be called \emph{fan
vertices}, and vertices of the form $(v,e,1,1)$, $(v,e,1,r)$, $(v,e,k,1)$,
$(v,e,k,r)$ will be called \emph{corner vertices}.

The edge set of $H(G,k,W_1,W_2)$ consists of the following pairs:
\begin{itemize}
  \item $(v,e,i,p)(v',e,i',p)$ such that $|i-i'|=1$;
  \item $(v,e,i,p)(v,e',i,p')$ such that $|p-p'|=1$;
  \item $u(v,e,1,1)$ for $u\in S_1\sse K_1$ and $(v,e,1,1)\in H_1$, where $S_1$
    is an arbitrary subset of $K_1$ whose cardinality is such that the degree of
    $(v,e,1,1)$ is exactly $W_1$. (Here we used that $W_1>2(n+m)$, as
    $(v,e,1,1)$ can have at most $n+m$ neighbours outside of $K_1$.)\\
Similarly, $u(v,e,1,r)$, $u(v,e,k,1)$, $u(v,e,k,r)$, $u(v,e,1,3)$, 
    $u(v,e,1,r-3)$, $u(v,e,k,3)$, $u(v,e,k,r-3)$, $u(v,e,3,1)$, $u(v,e,4,r)$, 
    $u(v,e,k-2,1)$, $u(v,e,k-3,r)$ for $u\in S_j\sse K_j$ (for
    $j=2\zd 12$ in this order) 
    and $(v,e,1,r)$, $(v,e,k,1)$, $(v,e,k,r)$, $(v,e,1,3)$, $(v,e,1,r-3)$,
    $(v,e,k,3)$, $(v,e,k,r-3)$, $(v,e,3,1)$, $(v,e,4,r)$, $(v,e,k-2,1)$,
    $(v,e,k-3,r)\in H_1$, where $S_2\zd S_{12}$ are arbitrary subsets whose cardinality is such that the degree of $(v,e,1,r),(v,e,k,1),(v,e,k,r)$ is exactly $W_1$ and the degree of the remaining vertices from the list is exactly $W_2$.
\end{itemize}

We study homomorphisms from $L(k,r,\ell_1,\ell_2)$ to $H(G,k,W_1,W_2)$. 
A homomorphism $\vf:L(k,r,\ell_1,\ell_2)\to
H(G,k,W_1,W_2)$ is said to be corner-to-corner (or c-c for short) if
\[
  \vf(1,1), \vf(1,r), \vf(k,1), \vf(k,r)\in \{(v,e,1,1), (v,e,1,r),(v,e,k,1),(v,e,k,r)\mid v\in V, e\in E\}.
\]
A homomorphism $\vf$ is
called identity (skew identity) if $\vf(i,p)\in\{(v,e,i,p)\mid v\in V,e\in E\}$
(respectively, $\vf(i,p)\in\{(v,e,k-i+1,p)\mid v\in V,e\in E\}$) for all
$i\in[k]$ and $p\in[r]$. Sometimes we will abuse the terminology and call 
a (skew) identity homomorphism the restriction of such homomorphism to
$L_1$ (the set of grid vertices).

We define the \emph{weight} of a homomorphism $\vf$ from 
$L(k,r,\ell_1,\ell_2)$ restricted to $L_1$ (the set of grid vertices) to 
$H(G,k,W_1,W_2)$ as the number of extensions of $\vf$ to
a homomorphism from $L(k,r,\ell_1,\ell_2)$. 

\subsection{Weights of Homomorphisms}
\label{subsec:weights}

We start with a simple lemma.

\begin{lemma}\label{lem:c-c-number}
The weight of an identity or skew identity homomorphism is
$W_1^{4\ell_1}W_2^{8\ell_2}$. 
\end{lemma}

\begin{proof}
Let $\vf$ be an identity or a skew identity homomorphism. 
The images of vertices from $L_1$ under $\vf$ are fixed,
while vertices from $L_2$ can be mapped by $\vf$
to any neighbour of the corresponding fan vertex independently. 
Since the degree of a corner vertex $(v,e,i,p)$ with $i\in\{1,k\}$
and $p\in\{1,r\}$ is $W_1$, and the degree of any other fan vertex 
is $W_2$, the result follows.
\end{proof}

The next lemma, which will be proved using Lemma~\ref{lem:c-c-number}, is essentially~\cite[Lemma~3.1]{Dalmau04:side} adapted to our setting, which in turn builds on~\cite[Lemma~4.4]{Grohe07:jacm}.

\begin{lemma}\label{lem:c-c-factorial}
Let $N$ be the number of $k$-cliques in $G$. Then the total weight of identity and skew identity homomorphisms is $2NW_1^{4\ell_1}W_2^{8\ell_2}k!$. 
\end{lemma}

\begin{proof}
We will show that the total weight of identity homomorphisms 
is $NW_1^{4\ell_1}W_2^{8\ell_2}k!$. Observe that there is a bijection 
between the sets of identity and skew identity homomorphisms that 
maps an identity homomorphism $\vf$ to a skew identity one, $\psi$,
for which $\psi(i,p)=(v,e,k-1+1,p)$ whenever $\vf(i,p)=(v,e,i,p)$. 
Therefore the total weight of skew identity homomorphisms is also 
$NW_1^{4\ell_1}W_2^{8\ell_2}k!$. 
First we give a description of all identity homomorphisms.
Let $\vc vk$ be the vertex set of a $k$-clique in $G$. 
For $p\in[r]$ with   $\vr(p)=\{a,b\}$, let $e_p=v_av_b$ be the edge in $G$ 
between $v_a$ and $v_b$. We define $\vf_{\vc vk}:L_1\to H_1$ by
\[
\vf_{\vc vk}((i,p))=(v_i,e_p,i,p)
\]
for every $i\in[k]$ and $p\in[r]$. 
  
We will need two claims; the first one follows directly from the definition.

\medskip

\noindent {\bf Claim 1.} 
$\vf_{\vc vk}$ is an identity homomorphism from $L_1$ to $H_1$.

\medskip

\noindent {\bf Claim 2.}
If $\vf$ is an identity homomorphism from $L_1$ to $H_1$ then
$\vf=\vf_{\vc vk}$ for some vertex set $\vc vk$ of a $k$-clique in $G$. 

\medskip

\noindent {\bf Proof of Claim 2.}
Let $\vf$ be an identity homomorphism from $L_1$ to $H_1$. 
  
For every $i\in[k]$ and $p\in[r]$, we have $\vf((i,p))=(v,e,i,p)$ for
some $v\in V$ and $e\in E$ with $v\in e \iff i\in p$. Let
$\vf((i',p))=(v',e',i',p)$. We claim that (A) $e=e'$. We prove (A) for $i'=i+1$,
the rest follows by induction. Since $\vf$ is a homomorphism and $(i,p)(i',p)$
is an edge in $L(k,r,\ell_1,\ell_2)$, there is an edge $(v,e,i,p)(v',e',i',p)$
in $H(G,k,W_1,W_2)$. The definition of edges in $H(G,k,W_1,W_2)$ implies that
$e=e'$. Similarly, let $\vf((i,p))=(v,e,i,p)$ and $\vf((i,p'))=(v',e',i,p')$. We claim
that (B) $v=v'$. For $p'=p+1$, this again follows from the assumption that $\vf$ is a homomorphism
and the definition of edges in $H(G,K,W_1,W_2)$; a simple induction establishes (B) for
arbitrary values $p,p'\in[r]$.

Together, claims (A) and (B) imply that there are vertices $\vc vk\in V$ and
edges $\vc er\in E$ such that for all $i\in[k]$ and $p\in[r]$ we have
$\vf((i,p))=(v_i,e_p,i,p)$. Since $\vf((i,p))\in H_1$, we have $v_i\in e_p \iff
i\in p$. Hence $\vc vk$ forms a $k$-clique in $G$. 
{\bf (End of proof of Claim 2.)}

\medskip

Claims~1 and 2 give us a complete description of identity homomorphisms from
$L_1$ to $H_1$: a mapping $\vf$ from $L_1$ to $H_1$ is an identity homomorphisms
if and only if $\vf=\vf_{\vc vk}$ for some vertex set $\vc vk$ of a $k$-clique in
$G$. Hence, the number of such mappings is the number of $k$-cliques in $G$
multiplied by $k!$. By Lemma~\ref{lem:c-c-number}, each identity homomorphism
can be extended in $W_1^{4\ell_1}W_2^{8\ell_2}$ distinct ways to a
homomorphism from $L(k,r,\ell_1,\ell_2) $ to $H(G,k,W_1,W_2)$. 
\end{proof}

We will frequently use the following simple observation. 

\begin{observation}\label{obs:par}
Let $\vf$ be a
homomorphism from a bipartite graph $G$ to a bipartite graph $H$. If vertices $u,v$ are of
distance $m$ in $G$ then $\vf(u),\vf(v)$ are of distance at most $m$ in $H$ and
the parity of the distances is the same.
\end{observation}

Next we establish an upper bound on the total weight of homomorphisms that 
are neither identity nor skew identity.

\begin{lemma}\label{lem:non-c-c}
Let $G=(V,E)$ have $n=|V|$ vertices and $m=|E|$ edges, let $k=4k'$ for 
some $k'$, and let $T=\log_{W_2}W_1$. If 
\[
\ell_1>\frac{8T\ell_2}{T-1},
\]
then the total weight of
homomorphisms that are neither identity nor skew identity is at most
\[
W_1^{4\ell_1}W_2^{6\ell_2}(2n+m)^{2\ell_2}
\cdot (4W_1+8W_2+nmkr)^{kr}.
\]
\end{lemma}

The key ideas in the proof of Lemma~\ref{lem:non-c-c} are the following: Firstly, we show
that c-c homomorphisms dominate non-c-c homomorphisms. Secondly, using crucially
the special structure of fan grids and our choice of $k$ being a multiple of
four, we establish an upper bound on any c-c
homomorphism that is neither identity nor skew identity. Finally, we give an
upper bound on the number of all homomorphisms. These three ingredients together
allows us to establish the required bound.

\begin{proof}
We prove this lemma in two steps. First, in Claims 1 and 2, we upper bound the
weight of a homomorphism that is not identity or skew identity. Second, in
Claim 3, we upper bound the number of such homomorphisms. 

\medskip

\noindent
{\bf Claim 1.}
The weight of any c-c homomorphism is greater than the weight of any 
non c-c homomorphism.

\medskip

\noindent
{\bf Proof of Claim 1:}
The weight of a c-c homomorphism $\vf$ is lower bounded by $W^{4\ell_1}_1$,
since each of the $4\ell_1$ neighbours of a corner vertex in
$L(k,r,\ell_1,\ell_2)$, say $u$, can be mapped to any of the $W_1$ neighbours
of the corner vertex $\vf(u)$ in $H(G,k,W_1,W_2)$.

The weight of any non c-c homomorphism is upper bounded by
$W_1^{3\ell_1+8\ell_2}W_2^{\ell_1}$, since in a non c-c homomorphism $\vf$ at
least one corner vertex in $L(k,r,\ell_1,\ell_2)$, say $u$, is mapped to a fan
vertex $\vf(u)$ that is not a corner vertex and hence the $\ell_1$ neighbours
of $u$ can only be mapped to the $W_2$ neighbours of $\vf(u)$ in
$H(G,k,W_1,W_2)$. The term $W_1^{3\ell_1+8\ell_2}$ corresponds to all but one
fan vertices in $L(k,r,\ell_1,\ell_2)$ being mapped to corner vertices in
$H(G,k,W_1,W_2)$.

Take the logarithm base $W_2$
of the two numbers above. We need to show that
\[
4T\ell_1> T(3\ell_1+8\ell_2)+\ell_1,
\]
or, equivalently,
\[
T\ell_1>8T\ell_2+\ell_1,
\]
which is equivalent to the condition
\[
\ell_1>\frac{8T\ell_2}{T-1}
\]
of the lemma.
{\bf (End of proof of Claim 1.)}

\medskip

\noindent{\bf Claim 2.}
Let $\vf$ be a c-c homomorphism that is neither identity nor skew 
identity. Then its weight does not exceed 
$W_1^{4\ell_1}W_2^{6\ell_2}(2n+m)^{2\ell_2}$.

\medskip

\noindent
{\bf Proof of Claim 2:}
We consider several cases. First observe some symmetries in c-c homomorphisms.
If $\psi$ is the mapping of $H_1$ (the ``grid'' part of $H(G,k,W_1,W_2)$) 
mapping $(v,e,i,p)$ to $(v,e,k-i+1,p)$, then the weight of
$\psi\circ\vf$ equals that of $\vf$. Thus we may assume
$\vf(1,1)\in\{(v,e,1,1),(v,e,1,r)\mid v\in V,e\in E\}$, which gives {\sc Case
1} and {\sc Case 2} below, respectively.
Note that by the assumption that $k$ is a multiple of four, both $k-1$ and
$r-1$, where $r={k\choose 2}$, are odd.

\medskip
  
\noindent {\sc Case 1.} $\vf(1,1)=(v,e,1,1)$ for some $v\in V,e\in E$.

\medskip

Since $(k,1)$ is at distance $k-1$ from $(1,1)$, by Observation~\ref{obs:par},
$\vf(k,1)$ is at odd distance not exceeding $k-1$ from $\vf(1,1)$. As $\vf$ is
c-c, there is only one possibility $\vf(k,1)=(v',e',k,1)$ for some $v'\in V,
e'\in E$. Similarly, as $(1,r)$ is at odd distance from $(1,1)$ and $\vf(1,r)$
is a corner vertex, by Observation~\ref{obs:par} it suffices to consider only two cases for $\vf(1,r)$.

\medskip

\noindent {\sc Case 1.1.} $\vf(1,r)=(v'',e'',1,r)$  for some $v''\in V, e''\in E$.

\medskip

Since $(k,r)$ is at distance $k-1$ from $(1,r)$, by Observation~\ref{obs:par},
$\vf(k,r)$ is at odd distance not exceeding $k-1$ from $\vf(1,r)$. As $\vf$ is
c-c and we assume that $\vf(1,r)=(v'',e'',1,r)$, there is only one possibility $\vf(k,r)=(v''',e''',k,r)$ for some $v'''\in
V, e'''\in E$. It is now easy to verify that $\vf$ is identity, a contradiction.

\medskip

\noindent {\sc Case 1.2.} $\vf(1,r)=(v'',e'',k,1)$  for some $v''\in V, e''\in E$.

\medskip

As in {\sc Case 1.1}, $\vf(k,r)=(v''',e''',1,1)$ for some $v'''\in V, e'''\in
E$. In detail, since $(k,r)$ is at distance $k-1$ from $(1,r)$, by
Observation~\ref{obs:par}, $\vf(k,r)$ is at odd distance not exceeding $k-1$
from $\vf(1,r)$. As $\vf$ is c-c and we assume that $\vf(1,r)=(v'',e'',k,1)$,
there is only one possibility $\vf(k,r)=(v''',e''',1,1)$.

Since $(1,r),(2,r)\zd(k,r)$ is
the only shortest path from $(1,r)$ to $(k,r)$, homomorphism 
$\vf$ maps this path to $(v_k,e_k,k,1),(v_{k-1},e_{k-1},k-1,1)
\zd(v_1,e_1,1,1)$ for some $\vc vk\in V, \vc ek\in E$ (in fact, we
can claim that $e'''=e_1=\dots=e_k=e''$, but we do not need 
this). In particular,
$\vf(4,r)=(v_{k-3},e_{k-3},k-3,1)$ and $\vf(k-3,r)=(v_4,e_4,4,1)$; that is,  
these two vertices are mapped to non-fan vertices. Since both 
$(v_{k-3},e_{k-3},k-3,1)$ and  $(v_4,e_4,4,1)$ have at most
$2n+m$ neighbours, the weight of $\vf$ is at most 
$W_1^{4\ell_1}W_2^{6\ell_2}(2n+m)^{2\ell_2}$.

\medskip

\noindent
{\sc Case 2.} $\vf(1,1)=(v,e,1,r)$  for some $v\in V,e\in E$.

\medskip

This case is symmetric to {\sc Case 1} so we do not give full details. Using
Observation~\ref{obs:par} and the assumption that $\vf$ is c-c, we get
$\vf(k,1)=(v',e',k,r)$  for some $v'\in V,e'\in E$. Also, we get that
$\vf(1,r)=(v'',e'',1,1)$ or $\vf(1,r)=(v'',e'',k,1)$ for some $v''\in V, e''\in
E$. In the former case, as in {\sc Case 1.1} we get that $\vf$ necessarily is
skew identity, which is a contradiction. In the latter case, similarly to 
{\sc Case 1.2}, we get that 
$\vf(3,1)=(v_{k-2},e_{k-2},k-2,r), \vf(k-2,1)=(v_3,e_3,3,r)$ 
 for some $v_{k-2},v_3\in V,e_{k-2},e_3\in E$. Since
$(v_{k-2},e_{k-2},k-2,r)$ and $(v_3,e_3,3,r)$ are not fan vertices, 
as in {\sc Case~1.2}, the weight of $\vf$
does not exceed $W_1^{4\ell_1}W_2^{6\ell_2}(2n+m)^{2\ell_2}$. 
{\bf (End of proof of Claim 2.)}

\medskip

\noindent
{\bf Claim 3.} The number of homomorphisms of the $(k\tm r)$-grid 
to $H(G,k,W_1,W_2)$ is upper bounded by 
\[
(4W_1+8W_2+nmkr)^{kr}.
\]

\medskip
\noindent
{\bf Proof of Claim 3:}
Since $H(G,k,W_1,W_2)$ has no more than $4W_1+8W_2+nmkr$ 
vertices and the $(k\tm r)$-grid has $kr$ vertices, the claim follows.
{\bf (End of proof of Claim 3.)}

\medskip

By Claims~1 and 2, the maximum weight of a homomorphism that is not 
identity or skew identity is 
$W_1^{4\ell_1}W_2^{6\ell_2}(2n+m)^{2\ell_2}$. By Claim~3,
there are at most $(4W_1+8W_2+kr)^{kr}$ such homomorphisms.
The result follows.
\end{proof}

We now have all results required to relate the number of $k$-cliques in a given
graph $G$ and the number of homomorphisms from $L(k,r,\ell_1,\ell_2)$ to
$H(G,K,W_1,W_2)$, for appropriately chosen values of $\ell_1,\ell_2,W_1,W_2$.

\begin{lemma}\label{lem:approx}
Let $N\geq 0$ be the number of $k$-cliques in $G$, where $k=4k'$ for 
some $k'$, $n=V(G)$, $m=E(G)$,
and $2n+m>6$. Let $M=M(\ell_1,\ell_2,W_1,W_2)$ be the number
of homomorphisms from $L(k,r,\ell_1,\ell_2)$, where $r={k\choose2}$, to 
$H(G,k,W_1,W_2)$. If $W_2=(2n+m)^2$, 
$W_1=W_2^2$, $\ell_2=8kr$, and $\ell_1=17\ell_2$, then we have
\[
N<\frac M{2W_1^{4\ell_1}W_2^{8\ell_2}k!}<N+\frac12.
\]
\end{lemma}

\begin{proof}
Let $M_c$ be the total weight of identity and skew identity homomorphisms 
and let $M_n$ be the total weight of the remaining homomorphisms. 
By Lemma~\ref{lem:c-c-factorial},
$M_c=2W_1^{4\ell_1}W_2^{8\ell_2}k!\cdot N$.
Therefore if $N\ge1$ we only need to show that 

\begin{equation}\label{eq:goal}
M_n<\frac{M_c}{2N}.
\end{equation}
Since $N\ge1$, 
\begin{equation}\label{eq:mc}
\frac{M_c}{2N}= W_1^{4\ell_1}W_2^{8\ell_2}k!\ge W_1^{4\ell_1}W_2^{8\ell_2},
\end{equation}
and it suffices to show that $M_n<W_1^{4\ell_1}W_2^{8\ell_2}$.
If $N=0$ then it again suffices to show that 
\[
M_n<W_1^{4\ell_1}W_2^{8\ell_2}k!\ge W_1^{4\ell_1}W_2^{8\ell_2}.
\]
On the other hand, $\ell_1=17\ell_2$ by the conditions of the lemma, that
is, $\ell_1>\frac{8T\ell_2}{T-1}$, where $T=\log_{W_2}W_1=2$. 
Therefore we satisfy the conditions of Lemma~\ref{lem:non-c-c}, and 
we have
\begin{equation}\label{eq:non-c-c1}
M_n<W_1^{4\ell_1}W_2^{6\ell_2}(2n+m)^{2\ell_2}
\cdot (4W_1+8W_2+nmkr)^{kr}.
\end{equation}
Note that for $n,m>0$, 
\begin{equation}\label{eq:nm}
8W_2=8(2n+m)^2<(2n+m)^4=W_1.
\end{equation}
Also, as $k\le n, r\le m$,
\begin{equation}\label{eq:kr}
nmkr<(2n+m)^4=W_1.
\end{equation}
Using~(\ref{eq:nm}) and~(\ref{eq:kr}) in~(\ref{eq:non-c-c1}), we get
\begin{equation}\label{eq:non-c-c2}
M_n<
W_1^{4\ell_1}W_2^{6\ell_2}(2n+m)^{2\ell_2}\cdot (6W_1)^{kr}.
\end{equation}
By~(\ref{eq:mc}) and~(\ref{eq:non-c-c2}), in order to 
establish~(\ref{eq:goal}) it suffices to prove
\begin{equation}
W_1^{4\ell_1}W_2^{6\ell_2} (2n+m)^{2\ell_2}\cdot (6W_1)^{kr}
< 
W_1^{4\ell_1}W_2^{8\ell_2},
\end{equation}
or, equivalently, that
\begin{equation}\label{eq:s1}
  (2n+m)^{2\ell_2}\cdot (6W_1)^{kr} < W_2^{2\ell_2}.
\end{equation}
  Since $(2n+m)^{2\ell_2}=W_2^{\ell_2}$ and $W_1=W_2^2$, 
  inequality~(\ref{eq:s1}) is equivalent to 
\begin{equation}\label{eq:s2}
  6^{kr}\cdot W_2^{2kr} < W_2^{\ell_2}.
\end{equation}
Since $2n+m>6$, we have 
\begin{equation}\label{eq:s3}
6^{kr}<(2n+m)^{kr}.
\end{equation}
Multiplying both sides of inequality~(\ref{eq:s3}) by $(2n+m)^{4kr}$, we obtain
\begin{equation}\label{eq:s4}
  6^{kr}\cdot(2n+m)^{4kr}<(2n+m)^{5kr}.
\end{equation}
Since $W_2=(2n+m)^2$, inequality~(\ref{eq:s4}) can be rewritten as
\begin{equation}\label{eq:s5}
  6^{kr}\cdot W_2^{2kr}<(2n+m)^{5kr}.
\end{equation}
Finally, since $W_2=(2n+m)^2$ and $\ell_2=8kr$,
(\ref{eq:s5}) implies~(\ref{eq:s2}).
\end{proof}

Finally, as Lemmas~\ref{lem:non-c-c} and~\ref{lem:approx} are only proved
for $k=4k'$, we need to show that the problem for other values of
the parameter can be reduced to $k$ of such form. The following lemma
takes care of that. Let $4p$-{\nclique} denote the following problem

\probfpt
{$4p$-\nclique}
{A graph $G$ and $k\in\zN$.}
{$k$.}
{The number of cliques of size $4k$ in $G$.}

\begin{lemma}\label{lem:to-good-k}
There is a parameterised AP-reduction from $p$-{\nclique} to $4p$-\nclique.
\end{lemma}

\begin{proof}
Let $G,k$ be an instance of $p$-{\nclique} and let $\ve\in(0,1)$ be an error tolerance.
If $k=4k'$ for some $k'$ then transform the instance to the instance 
$G,k'$ of $4p$-{\nclique} with the same error tolerance $\ve$. Otherwise 
repeat the following reduction as many times as required to obtain a 
parameter of the form $4k'$.

Suppose there is an FPRAS \textsf{Alg} that approximates the number of 
$(k+1)$-cliques in any graph. We construct graph $G^{+s}=(V',E')$ as 
follows. Let $\vc ws$ be vertices not belonging to $V$. Then set 
$V'=V\cup\{\vc ws\}$ and $E'=E\cup\{vw_i\mid v\in V, i\in[s]\}$,
that is, we connect all the new vertices with all vertices of $G$. The
following claim is easy to verify.

\medskip
\noindent
{\bf Claim 1.}
Let $N$ be the number of $k$-cliques in $G$ and let $N_1$ be the number 
of $(k+1)$-cliques in $G$. Then the number of $(k+1)$-cliques in 
$G^{+s}$ is $sN+N_1$. 

\medskip

Observe also that $N_1< nN$, because every $(k+1)$-clique contains 
a $k$-clique, and for every $k$-clique $C$ the number of $(k+1)$-cliques 
containing $C$ is at most $n-k$. Finally, we need the following observation.

\medskip
\noindent
{\bf Claim 2.}
In an instance $G,\ve$ of $k$-{\nclique}, the number $N$ of $k$-cliques
of $G$ can be assumed to be either 0 or greater than $3/2\ve$.

\medskip
\noindent
{\bf Proof of Claim 2.}
We show that there is a reduction from the general $k$-{\nclique} to
the probem admitting only instances with the restriction described in
Claim~2. The reduction makes use of the standard idea of blowing
up the vertices of $G$. Let $t$ be a natural number with 
$t>\left(\frac3{2\ve}\right)^{1/k}$. Construct $G^{(k)}$ by replacing
every vertex $v$ of $G$ with $\vc vt$, and every edge $vw$ with a 
complete bipartite graph on the vertices $\vc vt$, $\vc wt$. It is easy 
to see that every $k$-clique $v^1,\dots,v^k$ in $G$ gives rise to $t^k$ 
$k$-cliques in $G^{(k)}$ of the form $v^1_{i_1},\dots,v^k_{i_k}$.
Moreover, every $k$-clique of $G^{(k)}$ is of this form. Therefore
the number of $k$-cliqes in $G^{(k)}$ equals $t^kN$. By the choice 
of $t$
\[
t^kN > \left(\frac3{2\ve}\right)^{1/k\cdot k} N,
\]
and so if $N>0$, this number is greater than $3/2\ve$.
{\bf (End of proof of Claim~2.)}

\medskip

The reduction works as follows: Apply \textsf{Alg} to the instance 
$G^{+s},k+1$, where $s=\frac{3n}\ve$, with error tolerance $\ve/3$. 
If it returns a number 
$M$ output $\lfloor Q\rfloor$, where $Q=\frac Ms$. We now show that 
$(1-\ve)N<\lfloor Q\rfloor<(1+\ve)N$. By Claim~1 we have
\[
\left(1-\frac\ve3\right)(sN+N_1)<M<\left(1+\frac\ve3\right)(sN+N_1),
\]
or equivalently (by dividing by $s$), 
\[
\left(1-\frac\ve3\right)\left(N+\frac{N_1}s\right)<Q<\left(1+\frac\ve3\right)\left(N+\frac{N_1}s\right).
\]
Since by Claim~2 we assume that $N>3/2\ve$, we obtain 
\[
(1-\ve)N=\left(1-\frac\ve3\right)N-\frac{2\ve}3N<
\left(1-\frac\ve3\right)N-1<
\left(1-\frac\ve3\right)\left(N+\frac{N_1}s\right)-1,
\] 
implying $(1-\ve)N<\lfloor Q\rfloor$.

On the other hand, we have $N_1<nN$ and therefore
\[
\lfloor Q\rfloor\leq Q<\left(1+\frac\ve3\right)\left(N+\frac{N_1}s\right)<
\left(1+\frac\ve3\right)\left(N+\frac\ve3N\right)<(1+\ve)N,
\]
where in the middle inequality we used the choice of $s$. The result follows.
\end{proof}

In particular, Lemma~\ref{lem:to-good-k} establishes \#W[1]-hardness of the {$4p$-\nclique} problem.

\subsection{Putting the Pieces Together}
\label{subsec:pieces}

\begin{proof}[Proof of Theorem~\ref{thm:main}]
  As we mentioned earlier, conditions (1) and (4) are equivalent by
  Theorem~\ref{thm:exact} and the implications ``(1) $\Rightarrow$ (2)
  $\Rightarrow$ (3)'' are trivial.  

  The rest of the proof establishes ``$(3) \Rightarrow (4)$''. Assume that
  $\NCSP(\cC, -)$ admits an FPTRAS for a fan class $\cC$.
  Our goal is to show that $\cC$ has bounded treewidth. For the sake of contradiction,
  assume that $\cC$ has unbounded treewidth. We will exhibit
  a parameterised reduction from $p$-{\nclique} to $p$-$\NCSP(\cC,
  -)$, which gives an FPTRAS for $p$-{\nclique} assuming 
  an FPTRAS for
  $p$-$\NCSP(\cC, -)$. Under the assumption that FPT $\neq$ W[1] (under
  randomised parameterised reductions~\cite{Downey98:tcs}), the W[1]-hardness of
  $p$-{\clique} established in~\cite{DF95:fpt} implies,
  by~\cite[Corollary~3.17]{Meeks16:dam}, the non-existence of an 
  FPTRAS for the $p$-{\nclique} problem, a contradiction.

  Let $G=(V,E)$ and $k$ be an instance of the $p$-{\nclique} problem. 
  By Lemma~\ref{lem:to-good-k}, we can assume that $k=4k'$.
  First, we show that if $G$ has any $k$-cliques at all, it can be assumed 
to have many $k$-cliques. Let $s\in\nat$ and $G_s$ be defined as follows.
$V(G_s)=\{v_1\zd v_s\mid v\in V\}$ and $v_iw_j\in E(G_s)$, for 
$v,w\in V$ and $i,j\in[s]$, if and only if $vw\in E$. In other words,
every vertex $v$ of $G$ is replaced with $s$ distinct vertices $v_1\zd v_s$,
and every edge $vw$ is replaced with a complete bipartite graph $K_{s,s}$.

\medskip
\noindent
{\bf Claim 1.}
If $N$ is the number of $k$-cliques in $G$, then $G_s$ contains $s^kN$
$k$-cliques.

\medskip
\noindent
{\bf Proof of Claim 1.}
As is easily seen, for any indices $\vc ik\in[s]$ the vertices 
$v^1_{i_1}\zd v^k_{i_k}$ induce a clique in $G_s$ if and only if 
$v^1\zd v^k$ is a clique in $G$. Moreover, no clique in $G_s$ contains
vertices $v_i,v_j$ for $v\in V$ and $i,j\in[s]$. The result follows.
{\bf (End of proof of Claim~1.)}

\medskip

For a given instance $G=(V,E)$, $k$ of $p$-{\nclique} and error 
tolerance $\ve\in(0,1)$ using Claim~1, we first reduce it to the instance
$G_s$, $k$ of $p$-{\nclique}, where 
\[
s>\left(\frac{1+\ve/2}{\ve}\right)^{\frac1k}.
\]
Such a choice of $s$ guarantees that if $G_s$ contains any $k$-clique,
the number of $k$-cliques it contains is at least $\frac{1+\ve/2}\ve$. For simplicity we
will have this assumption directly for $G$. We will also assume that if 
 $n=|V|$ and $m=|E|$, then $2n+m>6$.

Now we construct an instance $\A,\B$ of $p$-$\NCSP(\cC, -)$ such that
an $\ve/2$-ap\-pro\-xi\-ma\-tion of the number of homomorphisms from
$\A$ to $\B$ yields an $\ve$-ap\-pro\-xi\-ma\-tion of the number of 
$k$-cliques in $G$. Structures $\A,\B$ will be chosen to be (essentially)
$\A=L(k,r,\ell_1,\ell_2)$ and $\B=H(G,k,W_1,W_2)$, where the parameters 
$\ell_1,\ell_2,W_1,W_2$ are set according to Lemma~\ref{lem:approx}.

Since $\cC$ is a fan class and we assume that $\cC$ is not of bounded treewidth,
there is a structure $\A$ in $\cC$ such that
$L(k,r,\ell_1,\ell_2)$ is the Gaifman graph $G(\A)$ of $\A$. 

We enumerate the class $\cC$ until we find such an $\A$. First we argue that $\A$ can be assumed to be a $\tau$-structure where $\tau$ consists of a
single binary relation symbol; i.e., $\A$ is a graph and hence
$L(k,r,\ell_1,\ell_2)$ itself. Let $\A$ be a
$\tau$-structure whose Gaifman graph $G(\A)$ is $L(k,r,\ell_1,\ell_2)$. We show
how to construct a $\tau$-structure $\B$ whose Gaifman graph $G(\B)$ is
$H(G,k,W_1,W_2)$ such that the set of homomorphisms from $\A$ to $\B$ is
identical to the set of homomorphisms from $G(\A)$ to $G(\B)=H(G,k,W_1,W_2)$, where $W_1=(2n+m)^4$ and $W_2=(2n+m)^2$. 
The universe of $\B$ is the vertex set of $H(G,k,W_1,W_2)$. Let
$R\in\tau$ and take any $\tuple{x}\in R^\A$. Since
$L(k,r,\ell_1,\ell_2)$ does not contain triangles, $\tuple{x}$ consists
of at most two distinct elements, say $a,b\in A$. Let $I\subseteq[\ar{R}]$ be
the set of indices $i$ with $\tuple{x}[i]=a$. For every $u,v\in B$ with $uv$ an
edge in $H(G,k,W_1,W_2)$, we add (if it is not there already) to $R^\B$ the
tuples $\tuple{y}$ and $\tuple{z}$ defined by $\tuple{y}[i]=\tuple{z}[j]=u$ and $\tuple{y}[j]=\tuple{z}[i]=v$ for every $i\in I$ and $j\not\in I$. Now it is easy to see that a mapping $\varphi:A\to B$ is a homomorphism from $\A$ to $\B$ if and only if $\varphi$ is a homomorphism from $G(\A)$ to $G(\B)$.

Since the parameters $n,m,\ell_1,\ell_2,W_1,W_2$ satisfy the conditions of 
Lemma~\ref{lem:approx}, by that lemma we have
\begin{equation}\label{eq:M1}
  N<\frac M{2W_1^{4\ell_1}W_2^{8\ell_2}k!}<N+\frac12,
\end{equation}
  where $N$ is the number of $k$-cliques in $G$, which we want to approximate
  within $\ve$, and $M$ is the number of homomorphisms from $\A$ to $\B$, for which we have an
  FPTRAS by assumption. Let $Q=M/(2W_1^{4\ell_1}W_2^{8\ell_2}k!)$. 
 The FPTRAS for  $p$-$\NCSP(\cC, -)$ applied with error tolerance $\ve/2$
 produces a number $M'$ such that
 \begin{equation}\label{eq:M2}
   (1-\ve/2)M<M'<(1+\ve/2)M.
 \end{equation}
 We then return $\lfloor Q'\rfloor$, where 
 \[
 Q'=\frac{M'}{2W_1^{4\ell_1}W_2^{8\ell_2}k!}.
 \]
 It remains to show that $(1-\ve)N<\lfloor Q'\rfloor<(1+\ve)N$.
 On one hand, we have
 \[
 \lfloor Q'\rfloor>\lfloor(1-\ve/2) Q\rfloor\ge\lfloor(1-\ve/2)N\rfloor\ge(1-\ve)N,
 \]
 where the first inequality follows from~(\ref{eq:M2}) and the definitions of
 $Q$ and $Q'$, the second inequality follows from~(\ref{eq:M1}) and the
 definitions of $Q$ and $N$, and the third inequality is trivial provided $N$
 is large enough (which we can assume by Claim~2 from the proof of 
 Lemma~\ref{lem:to-good-k}).

 On the other hand, we have
 \[
 \lfloor Q'\rfloor \leq Q'<(1+\ve/2)Q<(1+\ve/2)\left(N+\frac12\right),
 \]
 where the first inequality is trivial, the second inequality follows
 from~(\ref{eq:M2}) and the third inequality follows from~(\ref{eq:M1}).

 Assume first that $N=0$. Then $Q'<\frac{1+\ve/2}2$, and 
 by the assumption $\ve<1$ we have $\lfloor Q'\rfloor=0$ as required.
 Otherwise by the assumption on the number of $k$-cliques in $G$,
 $N>\frac{1+\ve/2}\ve$; therefore
 \begin{eqnarray*}
 \lfloor Q'\rfloor &<& (1+\ve/2)\left(N+\frac12\right)=
 (1+\ve/2)N+\frac{1+\ve/2}2\\
 &<& (1+\ve/2)N+(\ve/2)N=(1+\ve)N.
\end{eqnarray*}
 
Observe that the reduction runs in time $f(k)\cdot \mathsf{poly}(n+m,\ve^{-1})$
and is a parameterised AP-reduction. Thus, the reduction gives an FPTRAS for $N$. Theorem~\ref{thm:main} is proved. 
\end{proof}

\section{Conclusions}
\label{sec:conclusion}

We do not know whether Theorem~\ref{thm:main} holds for all classes of
(bounded-arity) relational structures.

With more technicalities (but the
same ideas as presented here), one can weaken the
assumption on a fan class to obtain the same result (Theorem~\ref{thm:main}). In
particular, it suffices to require that there are polynomials $f_1,f_2,f_3,f_4$
such that for any parameters $k,r,\ell_1,\ell_2\in\nat$, $G(\cC)$ contains the
fan-grid $L(k',r',\ell'_1,\ell'_2)$, where $k'=f_1(k,r,\ell_1,\ell_2)\ge k$,
$r'=f_2(k,r,\ell_1,\ell_2)\ge r$, $\ell'_1=f_3(k,r,\ell_1,\ell_2)\ge\ell_1$,
$\ell_2'=f_4(k,r,\ell_1,\ell_2)\ge\ell_2$. This can be achieved by making use of
Lemma~\ref{lem:to-good-k} (as it would not be possible to test directly for cliques of all sizes) and by a
modification of the construction from Section~\ref{sec:construction} (to
accommodate for the fact that some fan-grids may not correspond to cliques due
to incompatible numbers).

\section*{Acknowledgements}

We would like to thank the anonymous referees of both the conference~\cite{Bulatov19:mfcs}
and this full version of the paper.

\bibliographystyle{plainurl}
\bibliography{bz-lhs-approx}

\end{document}